\documentclass[preprint,number,sort&compress,12pt]{elsarticle}

\usepackage[utf8]{inputenc}
\usepackage[T1]{fontenc}
\usepackage{amsmath,amssymb}
\usepackage{amsthm}
\usepackage{booktabs}
\usepackage{multirow}
\usepackage{placeins}
\usepackage{graphicx}
\usepackage{tikz}
\usetikzlibrary{arrows.meta,positioning}
\usepackage{microtype}
\usepackage{xcolor}
\usepackage[colorlinks=true,linkcolor=blue!60!black,citecolor=blue!60!black,urlcolor=blue!60!black]{hyperref}
\usepackage{enumitem}
\setlist{itemsep=2pt,topsep=4pt}
\theoremstyle{plain}
\newtheorem{proposition}{Proposition}
\theoremstyle{definition}
\newtheorem{definition}{Definition}

\newcommand{\graft}{\textsc{Graft}}
\newcommand{\dense}{\textsc{Dense}}
\newcommand{\dq}[1]{\ensuremath{d@#1}}

\newcommand{\halfwidth}{\textwidth}\newcommand{\pairgap}{\par\vspace{8pt}}
\newcommand{\src}[1]{ Source: #1.}
\newcommand{\appkernel}{\ref{app:kernel}}

\begin{document}

\begin{frontmatter}

\title{Building Navigable Graphs Without Search\\ in Three Composable Stages}

\author{\'Edgar Ch\'avez}
\ead{elchavez@cicese.edu.mx}
\affiliation{organization={Computer Science Department, CICESE},
              city={Ensenada}, country={M\'exico}}

\begin{abstract}
Navigable graphs can be built without searching for neighbors: partition the
data, evaluate every pair inside each part, and select each point's edges
from the candidates. We give such a construction in three separable stages
and show that the middle one decides the quality. The \emph{pool} is any
partition with a few memberships per point. The \emph{ending} turns a
point's candidates into out-edges; ours keeps a bounded heap, prunes by
occlusion with a per-corpus slack, and appends reverse edges, re-pruning only
where a list overflows. The \emph{spine} is any edge set, exempt from the
prune, that keeps the graph reachable from its entry; ours, half-space-proximal
edges over a random sample, routes monotonically to every sampled point and
replaces a spanning tree at $1/10$ to $1/500$ of its cost. The ending composes with
any partitioner: on PiPNN's own candidate pool it beats PiPNN's ending on
each of six corpora from $10^6$ to $10^8$ points, by $3$ to $14\%$ in
distance evaluations at equal recall, and with 60 to 120 memberships per
point the composed build matches or beats a full dense construction at
$k=10$ and $k=100$ on all six, in $0.5$ to $0.9$ of its build time,
deterministically. The analysis explains why. Once a pool is localised its
quality is set by the data: every pool built on GIST lands within $4\%$ of
the exact-$k$NN ceiling, and the pairs a block cover misses are predicted,
point by point, by the local clustering of the $k$NN graph, whose
zero-clustering tail sets the memberships a corpus needs and grows with
$n$. All code, patches and logs are public.
\end{abstract}


\end{frontmatter}

\section{Introduction}\label{sec:intro}

This paper gives a construction of navigable graphs that never searches.
Partition the points into overlapping blocks; evaluate every pair inside
each block with a dense kernel; keep, for each point, its $C$ nearest
candidates; prune them to at most $R$ out-edges by the occlusion rule with
a slack $\alpha$; append the reverse of every surviving edge and re-prune
only the lists that overflow; and union a spine --- an edge set exempt from
the prune that keeps the graph reachable from its entry; ours is
half-space-proximal edges over a random sample. Three
stages --- the pool, the ending, the spine --- and no stage chooses a
distance to evaluate from the result of a previous one. The construction is
deterministic, and two of its stages are components that attach to any
partitioner, including the fastest one published.

The results are these. On PiPNN's own candidate pool~\cite{rubel2026pipnn},
our ending beats PiPNN's ending on every corpus we measured, by $3$ to
$14\%$ in distance evaluations at equal recall, for $8$--$16\%$ more build
time (Section~\ref{sec:ending}). With 60 to 120 memberships per point the
composed build --- PiPNN's partition, our ending, our spine --- matches or
beats a full dense construction on all six corpora at $k=10$ and $k=100$,
in $0.51$--$0.86\times$ its build time; on Deep-100M it needs $6$, $13$ and
$6\%$ fewer evaluations than the dense build at recall $0.95$, $0.99$ and
$k=100$ in $0.60\times$ its time, and $11$--$23\%$ fewer than Vamana
(Section~\ref{sec:regimes}). Our spine replaces the spanning tree of the
original dense build at $1/10$ to $1/500$ of its distance evaluations, is
never more than $1.1\%$ worse and up to $5\%$ better from $10^6$ to $10^8$,
and sizes itself (Section~\ref{sec:spine}). Every graph in the paper,
including the baselines, is graded by one searcher on one machine.

\emph{Why search-free.} Most graph indexes are built by search: HNSW
inserts points with a beam over the graph so far~\cite{malkov2020hnsw},
NSG and Vamana run a beam over a full graph and
prune~\cite{fu2019nsg,subramanya2019diskann}, ParlayANN batches the
insertions~\cite{manohar2024parlayann}, \graft{} harvests neighborhoods with
a beam over a frozen forest~\cite{chavez2026graft}. A beam's distance
evaluations are dependent random accesses, and a companion
paper~\cite{chavez2026wall} measures that this bounds construction an order
of magnitude below the machine's dense arithmetic, independently of the
algorithm and increasingly so with the dimension. Fixing the pairs before
evaluating them is the alternative, and it is older than the indexes it
competes with: recursive bisection with brute force in the
parts~\cite{chen2009recursive}, several random
partitions~\cite{wang2012scalable}, spanning trees inside random
clusters~\cite{munoz2019hcnng}, PiPNN's recursive random-leader partition
with a dense product per leaf~\cite{rubel2026pipnn}, and on GPUs the $k$NN
graphs of GGNN, CAGRA, cuVS and Tagore reworked into navigable
ones~\cite{groh2023ggnn,ootomo2024cagra,cuvs,li2025tagore}. These builders
are fast. Which of their decisions set the quality of the graph had not
been separated, and separating them is what makes two of the stages
reusable.

\emph{How quality is decided.} The second half of the paper measures each
stage with the other two held fixed. Where the partition is good the pool
is not the lever: with the ending fixed, every pool we built on GIST ---
$k$-means with five memberships, random leaders with thirty, a
shared-neighbor count --- lands within $4\%$ of the graph built from the
exact $k$NN lists (Section~\ref{sec:pool}). What a pool costs is written in
the data: the pairs a block cover misses are predicted, point by point and
in every corpus, by the local clustering coefficient of the $k$NN graph; the
fraction of points with zero clustering orders the corpora by difficulty,
grows with $n$ on the same corpus ($1.1\%$ at $2\cdot10^5$ to $8.1\%$ at
$10^8$ on Deep), and is what the 60 to 120 memberships above are paying
for. Uniform memberships are within $4\%$ of an oracle that knows each
point's clustering. The ending's own two decisions, the order of prune and
symmetrization and the slack against exact candidates, are each worth
$8$--$43\%$ on identical candidates.

\paragraph{Contributions}
\begin{enumerate}
\item \textbf{A reusable ending.} A bounded heap, an occlusion prune with a
  per-corpus slack, and prune-then-symmetrize order. On PiPNN's own pool it
  improves PiPNN's graphs on every corpus ($3$--$14\%$); on exact candidates
  it is $39$--$43\%$ better than Vamana's slack (Section~\ref{sec:ending}).
\item \textbf{A spine without search.} We define the spine by the
  reachability it must guarantee (Definition~\ref{def:spine}) and give an
  instance that needs no search: half-space-proximal edges over a random
  sample, which route monotonically to every sampled point
  (Proposition~\ref{prop:hsp}); with a degree cap of 16 and an automatic
  sample size they replace a spanning tree at $1/10$--$1/500$ of its cost
  (Section~\ref{sec:spine}).
\item \textbf{The composed build.} PiPNN's partition with our ending and
  spine matches or beats a full dense construction on six corpora at both
  $k$, in $0.51$--$0.86\times$ its time, and is bitwise deterministic at
  $10^6$ once three sources of nondeterminism in PiPNN are removed
  (Sections~\ref{sec:regimes} and~\ref{sec:ending}).
\item \textbf{Where quality comes from.} Once a pool is localised its
  quality is set by the ending, not the partitioner (GIST: every pool within
  $4\%$ of the exact-$k$NN ceiling); the cost of a pool is a property of the
  $k$NN graph's clustering, dominated by its zero-clustering tail, which
  grows with $n$; uniform memberships are near-optimal
  (Section~\ref{sec:pool}).
\end{enumerate}

\emph{Scope.} The kernels need distances that decompose into inner
products; for a black-box metric the evaluation count is the cost and
search-based builders are the right design. PiPNN builds $6$--$15\times$
faster than our dense construction, and the composed build is PiPNN's build
plus $8$--$16\%$: the contribution is quality per evaluation and a
component that composes, not a faster partitioner. The clustering result is
correlational evidence over six corpora and three values of $n$
(Section~\ref{sec:pool} states it as a conjecture). Partition-then-brute-force
construction has a long history, cited above; this paper separates its
stages and measures them.

Section~\ref{sec:protocol} defines the construction and the protocol.
Section~\ref{sec:regimes} reports the composed build against the dense
build, Vamana, PiPNN and HNSW. Sections~\ref{sec:ending},~\ref{sec:spine}
and~\ref{sec:pool} measure the ending, the spine and the pool in turn.
Section~\ref{sec:related} places the work, Section~\ref{sec:limits} states
its limits.

\section{The construction}\label{sec:protocol}

The construction takes $n$ points and produces a directed graph of
out-degree at most $R$ with an entry point, in three stages that can be
run, measured and replaced independently. This section defines each
stage, the configuration we recommend, the dense instance and its
kernel, the determinism guarantee, and the protocol under which every
number in the paper was measured.

\subsection{Three stages}

No stage chooses a distance to evaluate based on the result of a previous
one.

\begin{enumerate}
\item \textbf{Pool.} A partition into possibly overlapping blocks. A point's
  \emph{pool} is the union of its blocks; all pairs inside each block are
  evaluated. \dense{} uses $L$ leaders refined by a fixed number of Lloyd
  iterations~\cite{lloyd1982} on a seeded sample, and places each point in
  the leaves of its $m$ nearest leaders; PiPNN uses a recursive random-leader
  partition with $\approx 30$ leaves per point capped at 1,024
  points~\cite{rubel2026pipnn}.
\item \textbf{Ending.} Each point keeps the $C$ nearest of its candidates
  (\dense{}: $C=200$); the list is pruned by the occlusion rule with slack
  $\alpha$ to at most $R$ out-edges, the reverse of every surviving edge is
  appended to its target, and only lists that exceed $R$ are pruned again.
  PiPNN keeps, per leaf, each point's $\textit{mst\_deg}$ nearest (default
  2) in a hash-bucketed reservoir and prunes it with Vamana's $\alpha$.
\item \textbf{Spine.} An edge set unioned after the prune, and exempt from
  it, that guarantees the search can reach the parts of the graph the blocks
  leave disconnected (Definition~\ref{def:spine}). The stage is defined by that
  guarantee, not by a construction: the original \dense{} used one
  spatial-approximation tree (SAT)~\cite{navarro2002sat}; we use
  half-space-proximal (HSP) edges over a random sample
  (Section~\ref{sec:spine}).
\end{enumerate}

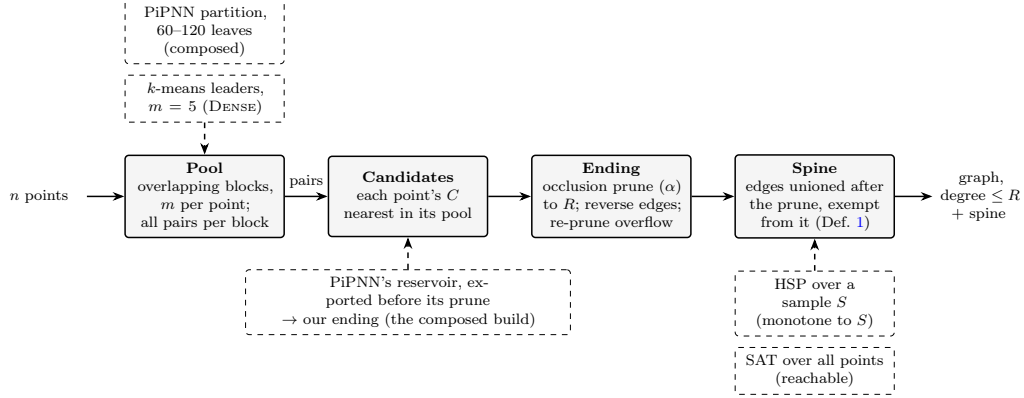
\begin{figure*}[!t]
\centering
\resizebox{\textwidth}{!}{%
\begin{tikzpicture}[font=\fontsize{8}{9.6}\selectfont,
  stage/.style={draw, rounded corners=2pt, thick, align=center, minimum height=15mm, text width=27mm, fill=black!4},
  inst/.style={draw, dashed, rounded corners=2pt, align=center, text width=27mm, minimum height=8mm},
  io/.style={align=center, text width=15mm},
  arr/.style={-{Stealth[length=2.2mm]}, thick}]
\node[io] (pts) {$n$ points};
\node[stage, right=7mm of pts] (pool) {\textbf{Pool}\\ overlapping blocks,\\ $m$ per point;\\ all pairs per block};
\node[stage, right=8mm of pool] (cand) {\textbf{Candidates}\\ each point's $C$\\ nearest in its pool};
\node[stage, right=8mm of cand] (end) {\textbf{Ending}\\ occlusion prune ($\alpha$)\\ to $R$; reverse edges;\\ re-prune overflow};
\node[stage, right=8mm of end] (spine) {\textbf{Spine}\\ edges unioned after\\ the prune, exempt\\ from it (Def.~\ref{def:spine})};
\node[io, right=7mm of spine] (g) {graph,\\ degree $\le R$\\ $+$ spine};
\draw[arr] (pts) -- (pool); \draw[arr] (pool) -- node[above]{pairs} (cand);
\draw[arr] (cand) -- (end); \draw[arr] (end) -- (spine); \draw[arr] (spine) -- (g);
\node[inst, above=6mm of pool] (km) {$k$-means leaders,\\ $m=5$ (\dense{})};
\node[inst, above=2mm of km] (pp) {PiPNN partition,\\ 60--120 leaves (composed)};
\draw[arr, dashed] (km) -- (pool);
\node[inst, below=6mm of spine] (hsp) {HSP over a sample $S$\\ (monotone to $S$)};
\node[inst, below=2mm of hsp] (sat) {SAT over all points\\ (reachable)};
\draw[arr, dashed] (hsp) -- (spine);
\node[inst, below=6mm of cand, text width=58mm] (exp) {PiPNN's reservoir, exported before its prune\\ $\to$ our ending (the composed build)};
\draw[arr, dashed] (exp.north) -- ([xshift=0mm]cand.south);
\end{tikzpicture}}
\caption{The construction. No stage chooses a pair to evaluate from the
result of a previous evaluation. Dashed boxes are interchangeable instances:
any partition can supply the pool (our $k$-means leaves, or PiPNN's partition
with its candidates exported before its own prune), and any edge set that
meets Definition~\ref{def:spine} can be the spine.}
\label{fig:pipeline}
\end{figure*}

\emph{Composition.} We patched PiPNN to export its reservoir, exactly
re-ranked, immediately before its final prune (read-only: its own graph is
unchanged); \dense{}'s ending then reads that file as its candidate lists.
This is the \emph{composed} build: PiPNN's pool, our ending and spine.

\emph{Recommended configuration.} For a new corpus: PiPNN's partition with
60 leaves per point and 10 nearest per leaf (120 where the corpus has a
heavy zero-clustering tail, Section~\ref{sec:pool}); the ending with
$C=200$, $R=64$, $\alpha=1$ on rich pools and $1.1$--$1.2$ on thin ones
(Section~\ref{sec:ending}); the HSP spine over a sample sized by the rule of
Section~\ref{sec:spine}, degree cap 16, with the orphan repair. This is the
configuration of Table~\ref{tab:memscale}.

\subsection{The dense instance and its kernel}

The \dense{} candidate stage evaluates every unordered pair once, in its
lowest shared leaf, offers the result to both endpoints' heaps, and runs as a
register-tiled symmetric scan scheduled in round-robin tournament order. On
the 64-thread Xeon used throughout it runs at $2.36$ G pair-evaluations per
second, $53\%$ of the machine's SGEMM rate, and as a fused GEMM-plus-select
on a 6 GB laptop GPU at $7.9$ G/s (\appkernel). With balanced
leaves the candidate stage costs $\tfrac12 n^2m^2/L$ pairs: quadratic at a
fixed fraction of the database, linear at a fixed absolute pool
$P = nm^2/L$, which is how it was run at $10^8$
($L = 24{,}000$, $P\approx 10^5$; the pair count grew $11.1\times$ for
$10\times$ the points).

\subsection{Determinism}

\begin{proposition}[Schedule independence]\label{prop:det}
Fix the data, the seed and the parameters. The \dense{} graph is independent
of the number of threads and of the order in which tiles execute.
\end{proposition}

\begin{proof}[Proof sketch]
The leaders are a deterministic function of the seeded sample and a fixed
summation order; memberships are the $m$ smallest
$(\text{distance},\text{index})$ keys; every pair is evaluated exactly once,
so the multiset of offers to each heap is fixed, and a heap retaining the $C$
smallest $(\text{distance},\text{id})$ keys is a commutative aggregation of
its offers; the prune is a per-vertex pure function of its list in that total
order; the HSP spine is a pure function of the seeded sample and its distance
matrix (Section~\ref{sec:spine}). A composition of pure functions is pure.
\end{proof}

The gate we run is stronger: an adjacency hash and the total distance count
must agree at 1 and 64 threads before any number is recorded; the distance
count has caught a silently skipped parallel loop that left the graph
byte-identical. Stock PiPNN is not deterministic (reruns differ by up to
$\pm3\%$ in query cost); Section~\ref{sec:ending} removes the causes.

\subsection{Protocol}

\emph{Machine.} d0, a 4-socket Xeon E7-4809 v3 (32 cores, 64 threads, AVX2,
1.5 TB), otherwise idle. \emph{Corpora.} GloVe-100 ($n = 1{,}183{,}514$,
cosine)~\cite{pennington2014glove}; SIFT-128 and GIST-960 ($10^6$,
$\ell_2$)~\cite{jegou2011pq}; Deep-96 at $10^7$ and $10^8$
($\ell_2$)~\cite{babenko2016deep1b}; Wikipedia, $6.35\cdot10^6$ BGE-M3
embeddings at $d=1024$ (cosine). \emph{Quality.} \dq{r} is the number of
distance evaluations per query that ParlayANN's beam~\cite{manohar2024parlayann}
needs to reach recall $r$ of the true $k$ nearest neighbors ($k=10$ unless
stated; $k=100$ where marked), interpolated on its beam ladder; every graph is
dumped and graded by this one searcher. QPS on the same ladder is reported
alongside but is not the cost metric: at equal \dq{r} Vamana's graphs run
faster (GloVe $17.9$k against $14.6$k QPS at $\approx 9{,}200$ evaluations),
for reasons --- degree, layout --- we have not measured.
\emph{Baselines.} ParlayANN's Vamana with the authors' recipe (R100/L200
on GloVe, R64/L128 $\alpha = 1.15$ on SIFT, $1.2$ elsewhere, two-pass);
PiPNN with its defaults; hnswlib M16, efC 200~\cite{malkov2020hnsw}, whose
level-0 graph we export and grade on the same searcher --- flat from one entry,
without the hierarchy, so its row is a \emph{lower bound} on hnswlib (its
native recall/QPS ladders are in the logs).

\emph{A methodological note.} An earlier round of our pool comparisons pruned
each point's \emph{whole} candidate row, where the dense build prunes only the
200 nearest. That one difference moved a shared-neighbor pool from
$5{,}510$ to $4{,}518$ evaluations at recall $0.95$ on GIST and reversed its
verdict. All pool comparisons below use the identical 200-nearest ending.

\section{The composed build against dense, Vamana, PiPNN and HNSW}\label{sec:regimes}

With 60 to 120 leaves per point the composed build matches or beats the
full dense construction on every corpus at $k=10$ and $k=100$, in
$0.51$--$0.86\times$ its build time; at PiPNN's default 30 leaves it already
beats PiPNN's own ending everywhere and the dense build on Wikipedia.
Table~\ref{tab:main} puts every system on one searcher at $k=10$ and
$k=100$ at default settings; Table~\ref{tab:builds} gives the build times of
the same runs; Tables~\ref{tab:glovemem} and~\ref{tab:memscale} give the
memberships that close the remaining cases.

\begin{table*}[!t]
\centering\small
\caption{\dq{0.95} at $k=10$ (top) and $k=100$ (bottom), ParlayANN's searcher,
d0, all systems at their default settings (more memberships:
Tables~\ref{tab:glovemem} and~\ref{tab:memscale}). \emph{Composed} = PiPNN's default pool $+$ our ending at the corpus's
slack (deterministic PiPNN at $10^6$--$10^7$; Wikipedia and Deep-100M are the
T1 runs: thin pool on Wikipedia, rich pool on Deep-100M). hnswlib is
its level-0 graph graded flat: a lower bound on hnswlib.\src{\texttt{open-items-FINDINGS.md} (K), \texttt{framing-FINDINGS.md} (T1, T5)}}
\label{tab:main}
\setlength{\tabcolsep}{4pt}
\begin{tabular}{@{}llrrrrr@{}}
\toprule
& corpus & \dense{}$+$HSP & Vamana & PiPNN & composed & hnswlib \\
\midrule
\multirow{6}{*}{\rotatebox{90}{$k=10$}}
& GloVe     & \textbf{9,172} & 9,204 & 13,330 & 11,444 & 12,919 \\
& SIFT      & 1,175 & \textbf{1,106} & 1,230 & 1,165 & 1,184 \\
& GIST      & \textbf{4,268} & 5,281 & 4,556 & 4,418 & 6,440 \\
& Deep-10M  & \textbf{2,133} & 2,150 & 2,431 & 2,277 & 2,278 \\
& Wikipedia & 3,840 & 4,246 & 4,158 & \textbf{3,750} & 4,225 \\
& Deep-100M & \textbf{3,677} & 4,047 & 4,289 & 4,108 & 4,133 \\
\midrule
\multirow{6}{*}{\rotatebox{90}{$k=100$}}
& GloVe     & \textbf{20,228} & 20,619 & 28,676 & 24,418 & --- \\
& SIFT      & 2,422 & 2,978 & 2,571 & \textbf{2,357} & 2,363 \\
& GIST      & 7,509 & 8,318 & 7,603 & \textbf{7,218} & 11,153 \\
& Deep-10M  & \textbf{4,084} & 4,807 & 4,761 & 4,596 & 4,580 \\
& Wikipedia & \textbf{8,030} & 8,666 & 8,688 & 8,713 & 9,629 \\
& Deep-100M & \textbf{7,198} & 7,369 & 8,403 & 7,795 & 8,790 \\
\bottomrule
\end{tabular}
\end{table*}

\begin{table*}[!t]
\centering\small
\caption{Build time in seconds, same runs as Table~\ref{tab:main} (\dense{}:
sum of its stages; PiPNN and Vamana: their reported build; composed: PiPNN's
build minus its prune plus ours; hnswlib M16, efC 200).}
\label{tab:builds}
\setlength{\tabcolsep}{5pt}
\begin{tabular}{@{}lrrrrr@{}}
\toprule
corpus & \dense{}$+$HSP & Vamana & PiPNN & composed & hnswlib \\
\midrule
GloVe     & 86 & 146 & 13 & 22 & 50 \\
SIFT      & 49 & 51 & 11 & 13 & 30 \\
GIST      & 194 & 546 & 41 & 56 & 182 \\
Deep-10M  & 689 & 699 & 114 & 143 & 401 \\
Wikipedia & 2,161 & 4,974 & 271 & 413 & 1,642 \\
Deep-100M & 13,189 & 9,816 & 1,220 & 984 / 1,785 & 5,684 \\
\bottomrule
\end{tabular}
\end{table*}

\paragraph{GloVe: the regime closes with memberships} With PiPNN's default
30 leaves per point the composed build trails the dense one by $25\%$. Keeping
more of each leaf (each point's 10, 20 or 50 nearest per leaf, or 2,048-point
leaves) recovers $5$--$9\%$; giving each point more leaves recovers the rest
(Table~\ref{tab:glovemem}). At 120 leaves per point the composed build beats
the dense build at recall $0.95$ and $0.99$ in $0.87\times$ its build time.
GloVe has the heaviest zero-clustering tail of the corpora
(Table~\ref{tab:trans}), and that tail is paid in memberships, as
Section~\ref{sec:pool} predicts --- not by the dense stage.

\begin{table}[!htbp]
\centering\small
\caption{GloVe: composed build against leaves per point (PiPNN's original
merge for the build times; one run each), our ending at $\alpha = 1$.\src{\texttt{open-items-FINDINGS.md} (G), \texttt{framing-FINDINGS.md} (T3)}}
\label{tab:glovemem}
\begin{tabular}{@{}lrrr@{}}
\toprule
leaves per point & \dq{0.95} & \dq{0.99} & build \\
\midrule
30 (PiPNN default) & 12,284 & --- & 17 s \\
60 & 9,971 & 25,400 & --- \\
90 & 9,539 & 24,306 & 60 s \\
120 & \textbf{8,956} & \textbf{23,832} & 75 s \\
\dense{} $+$ HSP & 9,172 & 24,833 & 86 s \\
\bottomrule
\end{tabular}
\end{table}

\paragraph{At $10^8$ and at $k=100$: the same remedy} With PiPNN's default
memberships the dense build still led in two places: on Deep-100M, where the
best composed build needed $12$--$13\%$ more evaluations at recall $0.95$ and
$7$--$11\%$ more at $0.99$, and at $k=100$, where the dense build was first on
GloVe, Deep-10M, Wikipedia and Deep-100M. Deep's zero-clustering tail grows to
$8.1\%$ at $10^8$ (Table~\ref{tab:deepscale}), so the mechanism that closed
GloVe should apply. It does (Table~\ref{tab:memscale}). On Wikipedia, 120
leaves per point give $2{,}888$ at $k=10$ and $7{,}185$ at $k=100$ against the
dense build's $3{,}840$ and $8{,}030$, at $0.77\times$ its time. On Deep-100M,
60 leaves already beat the dense build at every point measured --- $2\%$ at
recall $0.95$, $9\%$ at $0.99$, $2\%$ at $k=100$ --- at $0.52\times$ its time,
and Vamana by $11\%$ and $23\%$; 120 leaves widen the margin to $6\%$, $13\%$
and $6\%$ at $0.60\times$ the dense build's time. The two corpora where the
dense build led only narrowly close too: GIST at 60 leaves ($4{,}172$,
$8{,}682$ and $7{,}111$ against $4{,}268$, $11{,}098$ and $7{,}509$, in
$0.59\times$) and Deep-10M at 120 ($2{,}093$ and $4{,}345$ against $2{,}133$ and
$4{,}553$, level at $k=100$, in $0.84\times$). PiPNN's own ending
on the same rich pools does not improve ($4{,}437$--$4{,}511$ on Deep-100M
against $4{,}289$ at its default): the gain is the ending on a richer pool,
not the pool alone. The memberships needed follow the tail: 60--120 leaves
for Wikipedia ($2.4\%$) and Deep-100M ($8.1\%$), 120 for GloVe ($12.6\%$ at full $n$, $12.0\%$ in Table~\ref{tab:trans}).

\begin{table*}[!t]
\centering\small
\caption{Memberships: composed build (deterministic PiPNN at 60 or 120 leaves
per point, mst 10, our ending, HSP spine) against the dense build and Vamana.
One run each; Deep-100M's 120-leaf row was run twice and reproduces its
quality to the digit.\src{\texttt{memscale-FINDINGS.md}}}
\label{tab:memscale}
\setlength{\tabcolsep}{4pt}
\begin{tabular}{@{}llrrrr@{}}
\toprule
corpus & build & \dq{0.95} & \dq{0.99} & $k{=}100$ \dq{0.95} & time \\
\midrule
Wikipedia & \dense{} $+$ HSP & 3,840 & --- & 8,030 & 2,161 s \\
          & composed, 30 leaves & 3,750 & --- & 8,713 & 413 s \\
          & composed, 60 leaves & 3,218 & 20,034 & 7,681 & 830 s \\
          & composed, 120 leaves & \textbf{2,888} & 14,822 & \textbf{7,185} & 1,669 s \\
\midrule
Deep-100M & \dense{} $+$ HSP & 3,677 & 8,883 & 7,198 & 13,189 s \\
          & Vamana & 4,047 & 10,547 & 7,369 & 9,816 s \\
          & composed, 30 leaves & 4,108 & 9,516 & 7,795 & 1,785 s \\
          & composed, 60 leaves & 3,591 & 8,077 & 7,033 & 6,880 s \\
          & composed, 120 leaves & \textbf{3,449} & \textbf{7,741} & \textbf{6,772} & 7,885 s \\
\midrule
GIST      & \dense{} $+$ HSP & 4,268 & 11,098 & 7,509 & 194 s \\
          & composed, 60 leaves & \textbf{4,172} & \textbf{8,682} & \textbf{7,111} & 114 s \\
          & composed, 120 leaves & 4,386 & 9,377 & 7,353 & 190 s \\
\midrule
Deep-10M  & \dense{} $+$ HSP & 2,133 & 4,553 & \textbf{4,084} & 689 s \\
          & composed, 60 leaves & 2,157 & 4,586 & 4,194 & 349 s \\
          & composed, 120 leaves & \textbf{2,093} & \textbf{4,345} & 4,090 & 577 s \\
\bottomrule
\end{tabular}
\end{table*}

\paragraph{Build time} On Wikipedia (the most clustered corpus,
$2.4\%$ at zero) and SIFT the composed build matches or beats the dense one at
$k=10$ in $1/4$--$1/5$ of its time; GIST is within $4\%$ in $0.29\times$.

\section{The ending}\label{sec:ending}

The ending is where the quality of a search-free graph is decided, and it
is three decisions: a per-point heap bounded at $C$, an occlusion prune
with slack $\alpha$, and the order in which pruning and symmetrization are
applied. Two of them are worth $8$--$43\%$ on identical candidates
(Section~\ref{sec:order}); the ending as a whole improves PiPNN's graphs on
every corpus when applied to PiPNN's own pool (Section~\ref{sec:compose});
and it makes the composed build deterministic
(Section~\ref{sec:detcomposed}).

\subsection{Order and slack}\label{sec:order}

Two facts about the ending were established on identical candidates. First,
\emph{the order of prune and symmetrization}: symmetrizing every candidate and
pruning the union once, as \graft{} does, costs $8\%$ at recall $0.95$ and
$10\%$ at $0.97$ on GloVe against pruning each list first, appending reverse
edges and re-pruning only overfull lists ($9{,}910/14{,}387$ against
$9{,}130/13{,}031$; the prune stages are $21.6$ s against $10.4$ s). The
in-links of a pruned list are navigation edges another vertex selected;
re-pruning them under a different criterion discards half. Second,
\emph{slack against exact candidates}: Vamana's $\alpha = 1.15$ applied to 200
exact candidates on SIFT costs $39$--$43\%$ ($2{,}249/3{,}101$ against
$1{,}684/2{,}109$ at recall $0.98/0.99$); the slack compensates for candidates
a search may have missed, and there are none to compensate. We use $\alpha=1$
on dense pools.

\subsection{The ending composes with PiPNN}\label{sec:compose}

Table~\ref{tab:alpha} applies the ending to PiPNN's exported reservoir, on
PiPNN's default (\emph{thin}, $\approx 37$--$66$ candidates per point) and on
a richer one where each leaf keeps 10 nearest per point instead of 2
(\emph{rich}, $\approx 128$--$247$).

\begin{table*}[!t]
\centering\small
\caption{Our ending on PiPNN's pool, \dq{0.95}, against PiPNN's own ending on
the same run. HSP spine (Section~\ref{sec:spine}).\src{\texttt{framing-FINDINGS.md} (T2)}}
\label{tab:alpha}
\begin{tabular}{@{}lrrrrrrr@{}}
\toprule
& \multicolumn{4}{c}{thin pool} & \multicolumn{3}{c}{rich pool} \\
corpus & PiPNN & $\alpha\,1.0$ & $1.1$ & $1.2$ & PiPNN & $\alpha\,1.0$ & $1.2$ \\
\midrule
GIST     & 4,635 & 5,229 & \textbf{4,383} & 4,471 & 5,579 & \textbf{4,513} & 7,651 \\
GloVe    & 13,316 & 12,284 & 11,539 & \textbf{11,402} & 13,072 & \textbf{11,963} & 13,468 \\
SIFT     & 1,226 & \textbf{1,167} & 1,317 & 1,372 & 1,362 & \textbf{1,211} & 1,582 \\
Deep-10M & 2,462 & \textbf{2,278} & 2,416 & 2,406 & 2,399 & \textbf{2,328} & 3,059 \\
\bottomrule
\end{tabular}
\end{table*}

On every corpus some slack on PiPNN's \emph{own} pool beats PiPNN's ending:
GIST $-5\%$, GloVe $-14\%$, SIFT $-5\%$, Deep-10M $-7\%$, and at $10^8$ and on
Wikipedia $-3\%$ and $-10\%$ (Section~\ref{sec:regimes}). The slack is not a
function of pool size alone: on rich pools $\alpha=1$ always wins and
$\alpha=1.2$ is harmful (GIST $7{,}651$), while on thin pools the best value is
corpus-dependent --- $1.1$--$1.2$ on GIST and GloVe, $1.0$ on SIFT, Deep and
Wikipedia. The composed build costs PiPNN's build minus its prune plus ours:
$+8$--$16\%$ over the same PiPNN run.

\subsection{A deterministic composed build}\label{sec:detcomposed}

Our ending is deterministic (Proposition~\ref{prop:det}); PiPNN is not, and
neither is anything composed with it. We found three causes and removed them
behind a flag, leaving default PiPNN untouched: (i) the partition seed was
drawn from \texttt{std::random\_device}; (ii) merging undersized buckets used
a hash-based duplicate removal whose output order varies between runs, and
the next level samples leaders \emph{by position}, so the leaves themselves
changed; (iii) the sketch-hash vectors were drawn by a parallel tabulate
through one shared, stateful normal distribution, whose Box--Muller cache
made the hash family depend on scheduling. We also made the reservoir merge
order-free: candidates are totally ordered by (bfloat16 distance, id), a hash
bucket keeps its minimum, a full table evicts its maximum. The composed SIFT
graph is then byte-identical at 64 and 32 threads, at quality inside PiPNN's
run-to-run band (GIST $4{,}556$ against $4{,}548$--$4{,}672$). Keeping the
table's maximum cached across merge calls makes the order-free rule cost
$\approx 3\%$ of PiPNN's build even on 400-entry tables. At $10^8$ two runs
reproduce quality to the digit but not every byte (a $10^{-9}$ residual,
Section~\ref{sec:limits}).

\section{The spine}\label{sec:spine}

A graph built inside blocks can be disconnected, and the prune can make it
so even when the blocks overlap. This section says what a spine must
guarantee (Section~\ref{sec:spinedef}), compares the instances we built
against that guarantee (Section~\ref{sec:spineinst}), and then measures the
one we recommend, HSP edges over a sample (Section~\ref{sec:hsppractice}).

\subsection{What a spine must guarantee}\label{sec:spinedef}

A partition-built graph can shatter. On a synthetic set of 800 Gaussian
clusters, $k$-means leaves coincide with the clusters and an $\alpha=1$ prune
removes every cross-cluster candidate behind a nearer same-cluster one; the
graph is a disjoint union and returns recall $0.002$. More pool does not
help: the cross-leaf candidates are generated and then pruned. Whatever
repairs this must therefore be added \emph{after} the prune and be exempt
from it.

\begin{definition}[Spine]\label{def:spine}
Let $G=(V,E)$ be the graph after the prune, $e\in V$ its entry point and
$T\subseteq V$ a target set. A \emph{spine for $T$} is an edge set $E_S$,
unioned with $E$ after the prune and exempt from it, such that every $t\in T$
is reachable from $e$ in $(V, E\cup E_S)$. It is \emph{monotone for $T$} if,
from $e$ and from every vertex of $T$, greedy routing toward any $t\in T$
along $E_S$ strictly decreases the distance to $t$ at every step and ends at
$t$.
\end{definition}

The definition separates what a spine guarantees (its target set, and
whether reaching it is monotone) from how it is built and what it costs. Two
costs matter: the distance evaluations of its construction, and the degree it
adds to the vertices that carry it, which every beam expanding them pays
(Section~\ref{sec:hsppractice}).

\subsection{Instances}\label{sec:spineinst}

\begin{table}[!htbp]
\centering\footnotesize
\setlength{\tabcolsep}{3pt}
\caption{Spine instances against Definition~\ref{def:spine}. \emph{Recall}:
the clustered set, \dense{} L240/m5, $\mathit{ef}=64$, our searcher;
\emph{cost}: distance evaluations to build it on GloVe/GIST; $L$ leaders.\src{\texttt{adaptive-spine-FINDINGS.md}, \texttt{evp-FINDINGS.md}}}
\label{tab:spineinst}
\begin{tabular}{@{}lllrr@{}}
\toprule
instance & $T$ & guarantee & cost & recall \\
\midrule
none & --- & --- & 0 & 0.002 \\
$k$-nearest leaders & medoids & reachable & $L^2$ & 0.002 \\
\quad $+$ star & medoids & reachable & $L^2$ & 0.035 \\
SAT & $V$ & reachable & 0.99/1.7 G & 0.987 \\
HSP, $|S|=10$k & $S$ & monotone & 0.1 G & 0.993 \\
\quad cap 16 & $S$ & see text & 0.1 G & 0.968 \\
\bottomrule
\end{tabular}
\end{table}

Table~\ref{tab:spineinst} lists what we built.

\emph{Leader graphs} make the leaf medoids reachable, a target set at the
wrong granularity: the synthetic set has far more clusters than leaves, and
the disconnection is between clusters \emph{inside} a leaf. Even at the right
granularity, a beam does not expand a medoid farther from the query than its
frontier, so long edges hung on representatives are not taken.

\emph{The SAT} is a spanning tree over all points; its parent edges, added in
both directions, make every vertex reachable from the root ($T=V$). It is
the one adaptive stage of the original construction --- each point is routed
through the tree as it is built --- and it has a price: $985$ M distances on
GloVe and $1{,}678$ M on GIST, a root of out-degree $483$ and $1{,}133$, and
on GIST a graph that needs $5\%$ more evaluations than without it ($4{,}485$
against $4{,}259$ at recall $0.95$).

\emph{HSP over a sample.} Let $S$ be a seeded random sample containing the
entry point. For each $p\in S$, the half-space-proximal neighbors of $p$ in
$S$~\cite{chavez2005hsp} are found by scanning $S\setminus\{p\}$ in increasing
distance from $p$ and keeping $v$ unless some already-kept $u$ is closer to $v$
than $p$ is (Fig.~\ref{fig:hsp}). The spine is the union of these out-edges.

\begin{proposition}[Monotone routing to the sample]\label{prop:hsp}
For $p, t \in S$, $t \neq p$, either $t$ is an HSP neighbor of $p$ or $p$ has
an HSP neighbor $u$ with $d(u,t) < d(p,t)$. Hence HSP over $S$ is a spine
for $S$ that is monotone for $S$ (Definition~\ref{def:spine}).
\end{proposition}

\begin{proof}
If $t$ was not kept it was eliminated by a kept $u$ with $d(t,u) < d(t,p)$.
Each step strictly decreases $d(\cdot,t)$ over the finite set $S$, so greedy
routing ends at $t$; the entry is in $S$.
\end{proof}

\begin{figure}[!htbp]
\centering
\begin{tikzpicture}[scale=1.15, font=\fontsize{8}{9.6}\selectfont,
  pt/.style={circle, fill=black, inner sep=1.4pt},
  hollow/.style={circle, draw=black, fill=white, inner sep=1.4pt},
  arr/.style={-{Stealth[length=2mm]}, thick}]
\fill[black!8] (0.83125,-1.2) -- (2.9,-1.2) -- (2.9,1.7) -- (0.03125,1.7) -- cycle;
\draw[dashed] (0.83125,-1.2) -- (0.03125,1.7);
\node[font=\scriptsize, align=center] at (2.35,0.25) {closer to $u$\\ than to $p$};
\node[pt, label=below left:$p$] (p) at (0,0) {};
\node[pt, label=below:$u$] (u) at (1.0,0.25) {};
\node[pt, label=above left:$w$] (w) at (-0.8,1.1) {};
\node[hollow, label=right:$t$] (t) at (2.3,1.1) {};
\node[hollow, label=below:$v$] (v) at (1.45,-0.8) {};
\draw[arr] (p) -- (u); \draw[arr] (p) -- (w);
\draw[arr, densely dotted] (u) -- (t);
\end{tikzpicture}
\caption{HSP neighbors of $p$ within the sample. Scanning by distance from
$p$, $u$ is kept and eliminates every point closer to $u$ than to $p$ (shaded
half-plane), among them $v$ and $t$; $w$ is kept. Routing from $p$ toward $t$
steps to $u$, which is closer to $t$ than $p$ is
(Proposition~\ref{prop:hsp}).}
\label{fig:hsp}
\end{figure}
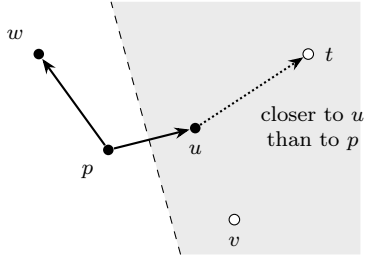

The guarantee is to the \emph{sample}, not to a query, and not to vertices
outside $S$: from the entry every region containing a sample point is
reached monotonically, and the sample point's own edges in $G$ carry the
search into its region. What makes that enough is coverage, and the
clustered set measures it: with $|S| = 1{,}000$ about $29\%$ of the 800
clusters hold no sample point ($e^{-1.25}$) and recall is $0.531$; with
$|S| = 10{,}000$ every cluster is hit and recall is $0.993$, above the tree's
$0.987$. The cost is $|S|^2$ distances, computed as a matrix and swept
row-parallel: no search.

\emph{Other instances} fit the definition and were not built: the monotonic
relative neighborhood graph of NSG~\cite{fu2019nsg} over $S$, which is also
monotone for $S$; a forest of shallow trees, which keeps $T=V$ while
dividing the root's degree; or the upper layers of a hierarchical graph built
over $S$, which would reintroduce search into the spine only.

\subsection{The HSP spine in practice}\label{sec:hsppractice}

\paragraph{The price of a spine is its degree} HSP degree is not bounded in
high dimension. The median out-degree on real data grows as a power of the
sample size up to $10^5$ with no saturation (Table~\ref{tab:hspdeg}), and the
sample points' extra out-edges are paid by every beam that expands them.
Capping each row at its first $c$ HSP neighbors (Table~\ref{tab:cap}) removes
the price: capped at 8 or 16, every sample size up to $30$k is within $\approx
1\%$ of no spine, where uncapped spines cost $+6\%$ at $30$k and $+18$--$26\%$
at $100$k. The cap has a floor set by the guarantee: at a cap of 8 the
clustered set falls to $0.66$ recall at $|S|=10$k, because the dropped far
neighbors are exactly the eliminators Proposition~\ref{prop:hsp} uses, so a
capped HSP spine is no longer monotone for $S$ by construction; at 16
it keeps $0.968$ at $10$k and $0.998$ at $30$k.

\begin{table*}[!t]
\begin{minipage}[t]{\halfwidth}
\centering\footnotesize
\caption{Median HSP out-degree against sample size (uncapped).\src{\texttt{adaptive-spine-FINDINGS.md} section 3}}
\label{tab:hspdeg}
\begin{tabular}{@{}lrrrrrr@{}}
\toprule
$|S|$ & 1k & 3k & 10k & 30k & 100k & fit \\
\midrule
GloVe ($d=100$) & 19 & 28 & 41 & 57 & 81 & $\propto |S|^{0.31}$ \\
GIST ($d=960$) & 9 & 12 & 17 & 23 & 32 & $\propto |S|^{0.27}$ \\
clustered ($d=128$) & 20 & 24 & 28 & 31 & --- & saturating \\
\bottomrule
\end{tabular}
\end{minipage}\pairgap
\begin{minipage}[t]{\halfwidth}
\centering\footnotesize
\caption{Spine price, \dq{0.95}, \dense{} build, by sample size and degree
cap. No spine: GloVe 9,130, GIST 4,259; SAT: 9,069 and 4,485.\src{\texttt{adaptive-spine-FINDINGS.md} section 3}}
\label{tab:cap}
\begin{tabular}{@{}lrrrr@{}}
\toprule
& \multicolumn{2}{c}{GloVe} & \multicolumn{2}{c}{GIST} \\
$|S|$ & uncapped & cap 16 & uncapped & cap 16 \\
\midrule
1k & 9,138 & 9,121 & 4,249 & 4,246 \\
10k & 9,337 & 9,172 & 4,331 & 4,268 \\
30k & 9,420 & 8,961 & 4,501 & 4,298 \\
100k & 10,743 & --- & 5,349 & --- \\
\bottomrule
\end{tabular}
\end{minipage}
\end{table*}

\paragraph{Sizing the sample without a cluster count} The cluster count is
unknown on real data. We start at $|S| = \max(1024, \sqrt n)$ and double until
a round adds fewer than $n/10^4$ vertices reachable from the entry, measured
by a breadth-first search over the final graph. A first version waited for
\emph{every} vertex to be reachable and never stopped on GIST: the prune
leaves some vertices with no in-edge at all --- $0.19\%$ on GIST, $0.009\%$ on
Deep-10M, none on GloVe and SIFT --- and a spine among sample points cannot
reach them. We give each such orphan an in-edge from its own nearest kept
out-neighbor, which costs no distance evaluation. The rule then picks
$1{,}024$--$3{,}162$ on the real corpora and $8{,}192$ on the clustered set,
at no-spine quality (GIST $4{,}290$, GloVe $9{,}135$, SIFT $1{,}085$,
Deep-10M $2{,}151$). The repair does not change quality, and it does not
explain one gap that remains: at the top of the beam ladder the HSP-spined
GIST graph reaches recall $0.9962$, the SAT-spined one $0.9977$.

\paragraph{Every dense table, $10^6$ to $10^8$} Table~\ref{tab:t5} replaces
the tree by the HSP spine ($|S| = 10$k, cap 16) in every \dense{} build.

\begin{table*}[!t]
\centering\footnotesize\setlength{\tabcolsep}{3pt}
\caption{\dense{} with SAT against HSP spine, \dq{0.95} and \dq{0.99}; spine
cost in distances.\src{\texttt{framing-FINDINGS.md} (T5)}}
\label{tab:t5}
\begin{tabular}{@{}lrrrrr@{}}
\toprule
corpus & SAT \dq{0.95} & HSP \dq{0.95} & SAT \dq{0.99} & HSP \dq{0.99} & spine, HSP/SAT \\
\midrule
GloVe & 9,069 & 9,172 ($+1.1\%$) & --- & 24,833 & 100 M/985 M \\
SIFT & --- & 1,175 & 2,155 & 2,116 ($-1.8\%$) & 100 M/--- \\
GIST & 4,485 & 4,268 ($-4.8\%$) & --- & 11,098 & 100 M/1,678 M \\
Deep-10M & 2,180 & 2,133 ($-2.2\%$) & 4,811 & 4,553 ($-5.4\%$) & 100 M/--- \\
Wikipedia & 3,868 & 3,840 ($-0.7\%$) & --- & --- & 100 M/--- \\
Deep-100M & 3,775 & 3,677 ($-2.6\%$) & 8,831 & 8,883 ($+0.6\%$) & 100 M/56.7 G \\
\bottomrule
\end{tabular}
\end{table*}

Never more than $1.1\%$ worse than the tree, $2$--$5\%$ better on four
corpora; on Deep-100M the build drops from 4:13 to 3:49 hours and its peak
memory from 327 to 267 GB. With the HSP spine the construction has no
adaptive stage at all.

\section{The pool}\label{sec:pool}

This section explains the results above from the data. Once a pool is
localised, which pairs it evaluates matters less than what the ending does
with them (Section~\ref{sec:notlever}); what a pool must cost to cover the
$k$NN graph is set by that graph's clustering, and the tail of unclustered
points is what the memberships of Section~\ref{sec:regimes} are buying
(Section~\ref{sec:clustering}).

\subsection{Where the partition is good, the pool is not the lever}\label{sec:notlever}

The cleanest test of a pool is to give the identical ending the best
possible one: each point's \emph{exact} $k$NN list. Table~\ref{tab:ceiling}
compares that ceiling with every pool we built on GIST.

\begin{table}[!htbp]
\centering\small
\caption{GIST-960: pools against the exact-$k$NN ceiling, identical ending
(200-nearest heap, $\alpha=1$, cap 64, SAT spine), ParlayANN's searcher.
\emph{Pool} is the mean number of candidates a point sees.\src{\texttt{capped-}, \texttt{random-}, \texttt{pipnn-FINDINGS.md}}}
\label{tab:ceiling}
\begin{tabular}{@{}lrr@{}}
\toprule
pool & pool/pt & \dq{0.95} \\
\midrule
exact $k$NN lists (ceiling) & --- & 4,601 \\
$k$-means L700, $m=5$ & 53k & 4,485 \\
$k$-means L2500, $m=5$ & 18k & 4,686 \\
random leaders L90000, $m=30$ & 46k & 4,419 \\
shared-neighbor count, top 4,000 & 4k & 4,518 \\
\midrule
Vamana (search-built) & --- & 5,281 \\
PiPNN (own pool and ending) & 14k & 4,548 \\
\bottomrule
\end{tabular}
\end{table}

Every pool lands between $4\%$ below and $2\%$ above the ceiling, three of them below it: on
GIST the quality is set by the ending, not by which pairs are evaluated.
Pool \emph{recall} still follows the measure of the region scanned --- on
GIST $0.58$ of true 10-NN pairs at a 20k pool, $0.72$ at 36k, $0.80$ at 50k,
whatever its shape (permutation tries, Voronoi cell neighborhoods,
HSP-adjacent cells, per-point memberships) --- but beyond $\approx 0.8$ the
ending cannot use more. GloVe is different: its ceiling is $8{,}652$ against
$9{,}069$ for the $k$-means pool, a $5\%$ gap, and its pool matters.

\subsection{The cost of a pool is a property of the $k$NN graph}\label{sec:clustering}

Covering the $k$NN graph $G_k$ with blocks is an edge-clique cover: a true
pair $(p,q)$ is covered if $p$ and $q$ share a block, and the cost is
$\sum_j |B_j|^2$. If $G_k$ were a disjoint union of cliques the cost would be
$\approx nk$. We measured how far each corpus is from that, on a sample of
2,000 points and the exact $k$NN lists of their neighbors, and the cost of a
symmetric $k$-means cover (cells of 1,400) as the number of memberships
$m@\rho$ needed to cover a fraction $\rho$ of true 10-NN pairs.

\begin{table*}[!t]
\centering\footnotesize
\caption{Transitivity of the exact 10-NN graph and the memberships a block
cover needs, all corpora subsampled to $n = 200{,}000$. \emph{cc} is the local
clustering coefficient.\src{\texttt{transitivity-FINDINGS.md}}}
\label{tab:trans}
\begin{tabular}{@{}lrrrrrrr@{}}
\toprule
corpus & LID & recip. & mean cc & cc $=0$ & $m@0.85$ & $m@0.90$ & $m@0.95$ \\
\midrule
Wikipedia & 35.3 & 0.517 & 0.162 & 2.4\% & 1.77 & 1.99 & 2.81 \\
SIFT      & 19.7 & 0.332 & 0.137 & 0.7\% & 2.01 & 2.48 & 2.95 \\
GIST      & 49.9 & 0.074 & 0.076 & 4.3\% & 2.95 & 3.53 & 4.38 \\
GloVe     & 35.7 & 0.312 & 0.106 & 12.0\% & 3.30 & 4.24 & 5.91 \\
\bottomrule
\end{tabular}
\vspace{6pt}
\centering\footnotesize
\caption{Deep at four sizes, fixed cell size 1,400.\src{\texttt{framing-FINDINGS.md} (T4)}}
\label{tab:deepscale}
\begin{tabular}{@{}rrrrrr@{}}
\toprule
$n$ & $L$ & mean cc & cc $=0$ & $m@0.85$ & $m@0.85/\ln L$ \\
\midrule
$2\cdot10^5$ & 143 & 0.144 & 1.1\% & 1.95 & 0.39 \\
$10^6$ & 714 & 0.122 & 2.7\% & 2.49 & 0.38 \\
$10^7$ & 7,143 & 0.096 & 5.7\% & 3.61 & 0.41 \\
$10^8$ & --- & 0.090 & 8.1\% & --- & --- \\
\bottomrule
\end{tabular}
\end{table*}

Three observations (Table~\ref{tab:trans}).
\emph{The overpay is structural}: mean clustering is $0.08$--$0.17$ and a
point's neighbors with their neighbors span $7.7$--$9.1$ neighborhoods of a
possible 11.
\emph{Clustering predicts which pairs are missed, point by point}: within
every corpus, a point's cover rises monotonically across clustering quartiles
(16 of 16 settings), and Spearman's correlation of cover with clustering
($+0.09$ to $+0.61$) exceeds that with reciprocity or local intrinsic
dimension in every case; on GloVe at $m=2$ the least-clustered quarter is
covered $0.32$, the most-clustered $0.86$.
\emph{Across corpora the tail decides}: the mean misorders GloVe and GIST,
the fraction with zero clustering does not, and the memberships needed at
high cover follow it.

\emph{The tail grows with scale.} On one corpus (Table~\ref{tab:deepscale})
the zero-clustering fraction rises from $1.1\%$ to $8.1\%$ over three decades
of $n$ and the memberships needed grow as $\approx 0.4\ln L$: scale makes a
corpus GloVe-like. Full Wikipedia ($6.35\cdot10^6$) is the most clustered
corpus we measured (mean $0.206$, $2.4\%$ at zero).

\paragraph{Uniform memberships are near-optimal} If the misses concentrate
on low-clustering points, give those points more memberships. We tested it
with the exact clustering as an oracle and with two build-time proxies
(clustering on a pilot $m=2$ graph; leader margin), over two-level policies
(the hardest fraction $f\in[0.05,0.5]$ gets $m_{hi}$, the rest $m_{lo}$). The
best policy at the uniform-$m=5$ pair budget saves $3$--$4\%$ of pairs on
GloVe and nothing on GIST; end to end every policy is inside the $\pm3\%$
band; and $22\%$ more pairs buy nothing. A hard point's extra memberships
enlarge those leaves for every member, and pairs are quadratic in leaf size:
the cost of covering the $k$NN graph is spread over the corpus.

\paragraph{The conjecture} The pairs a block cover needs for a given
cover $\rho$ are bounded below by a functional of the $k$NN graph's
clustering distribution, dominated by its zero-clustering tail. We have
correlational evidence over six corpora and three values of $n$, not a
proof.

\section{Related work}\label{sec:related}

\paragraph{Graph indexes built by search} From navigable small
worlds~\cite{malkov2014nsw} through HNSW~\cite{malkov2020hnsw},
NSG~\cite{fu2019nsg}, FANNG~\cite{harwood2016fanng} and
Vamana/DiskANN~\cite{subramanya2019diskann} to ParlayANN~\cite{manohar2024parlayann}
and \graft{}~\cite{chavez2026graft}, a beam supplies each point's candidates;
the cost model that bounds them is measured in~\cite{chavez2026wall}, after
Hong and Kung~\cite{hong1981io} and the roofline~\cite{williams2009roofline}.

\paragraph{Partition then brute force} Recursive Lanczos
bisection~\cite{chen2009recursive}, multiple random
divisions~\cite{wang2012scalable}, HCNNG~\cite{munoz2019hcnng}, PiPNN~\cite{rubel2026pipnn}
and graph-partitioned indexes~\cite{gottesburen2025partition} all evaluate
pairs inside parts; SPANN partitions for disk-resident
search~\cite{chen2021spann}. On GPUs, GGNN~\cite{groh2023ggnn},
CAGRA~\cite{ootomo2024cagra}, cuVS all-neighbors~\cite{cuvs},
Tagore~\cite{li2025tagore} and FAISS~\cite{johnson2021gpu,faiss} build $k$NN
graphs by blocked products and rework them. Our construction is a member of
this family; the contribution here is the separation of its stages and the
measurement of each, not the construction. ParlayANN's builders are
deterministic by design; what we add is the determinism of a composed build
with a partitioner that was not.

\paragraph{Refinement} NN-descent~\cite{dong2011nndescent},
EFANNA~\cite{fu2016efanna}, SPTAG~\cite{sptag2018} and RNN-Descent~\cite{ono2023rnndescent}
refine a graph by batched rounds; Yang et al.\ identify candidate acquisition
as the dominant construction cost and accelerate RNG pruning~\cite{yang2025revisiting}.

\paragraph{Pruning and connectivity} The occlusion rule descends from the
relative neighborhood graph~\cite{toussaint1980rng,arya1993rng}; the HSP
graph~\cite{chavez2005hsp} and the SAT~\cite{navarro2002sat} are the two
structures we use as spines. The intrinsic-dimension dependence of search
work is surveyed in~\cite{aumuller2021lid,chavez2001searching}; quantisers
that compress the scan side~\cite{jegou2011pq,guo2020scann,gao2024rabitq,connor2025evp}
are lossless filters inside a pool in our measurements and do not change its
quality (\appkernel).

\section{Scope and limitations}\label{sec:limits}

\begin{enumerate}
\item \textbf{Vectors, not metric spaces.} The kernels need distances that
  decompose into inner products.
\item \textbf{Six corpora, three values of $n$.} The clustering results are
  correlational; the conjecture of Section~\ref{sec:pool} is not proved.
\item \textbf{One run per PiPNN configuration} in the membership studies (two at
  $10^8$, 120 leaves), and
  PiPNN's stock build varies by up to $\pm3\%$; the deterministic variant
  removes the variance at about $3\%$ of PiPNN's build time.
\item \textbf{QPS is not explained.} At equal distance counts Vamana's graphs
  run faster; we report distances as the cost and QPS alongside.
\item \textbf{hnswlib is measured as a lower bound} (level-0 graph, flat).
\item \textbf{The HSP/SAT recall-ceiling gap on GIST} ($0.9962$ against
  $0.9977$) is unexplained.
\item \textbf{Memory.} \dense{}'s heaps are $nC$ keys of 8 bytes, 160 GB at
  $10^8$; PiPNN's reservoir at 400 entries is 320 GB at $10^8$, and the
  deterministic 120-leaf build peaked at 679 GB.
\item \textbf{Determinism at $10^8$ is not bitwise.} The composed build is
  byte-identical across thread counts at $10^6$ (gated); two runs at $10^8$
  reproduce quality to the digit but differ in $38$ of $3.8\cdot10^{10}$
  candidates and $528$ of $4.1\cdot10^9$ edges. The source is not identified.
\item \textbf{Beyond $10^8$.} A localiser over $2\cdot10^7$ random leaders was
  run at $10^9$ (9.6 h); no graph is built past $10^8$.
\end{enumerate}

\section{Conclusion}\label{sec:conclusion}

A navigable graph can be built without search, deterministically, from
three separable stages, and the stage that decides its quality is the one
that costs least: the ending. Attached to the fastest published partition it
improves that partition's graphs on every corpus, and with enough
memberships per point it matches or beats a full dense construction at
$10^6$ to $10^8$ points and at $k=10$ and $k=100$, in half to nine tenths
of the time. Connectivity needs a spine, defined by the reachability it
guarantees, not by a construction; the instance we give, half-space-proximal
edges over a sample, needs no search, costs $|S|^2$ distances and routes
monotonically to every sampled point. What remains with the data is the pool: how many
memberships a corpus needs is written in the tail of its $k$NN graph where
neighbors do not know each other, that tail grows with $n$, and spending
memberships uniformly is as good as spending them where it is. The
construction, the patches that make PiPNN deterministic and exportable,
and every log are public.

\section*{Reproducibility}

All code, patches, scripts and logs are public. The construction is the
\texttt{fg} binary of \texttt{graft-ann}, branch \texttt{fgraft}
(\url{https://github.com/zevahcle/graft-ann/tree/fgraft}; flags
\texttt{-{}-dense}, \texttt{-{}-dense-m}, \texttt{-{}-dense-C},
\texttt{-{}-dense-prune}, \texttt{-{}-dense-cands}, \texttt{-{}-dense-m-file},
\texttt{-{}-dense-spine 3} with \texttt{-{}-dense-spine-s}, \texttt{-{}-dense-spine-cap},
\texttt{-{}-dense-spine-repair}). The PiPNN patch (counters, reservoir export,
determinism), every chain that produced a number in this paper, its log and
the findings file that reads it are at
\url{https://github.com/zevahcle/fgraft-experiments}.
Every \dense{} graph is gated by an adjacency hash and a distance count at 1
and 64 threads; every graph is graded by ParlayANN's searcher.

\appendix

\section{The dense kernel}\label{app:kernel}

The \dense{} candidate stage is one hot loop. Table~\ref{tab:kernel} is the
ladder from the first working version to the one used throughout; every
step produced a byte-identical GloVe graph, which is what made it safe to
climb. Four changes account for the $3.4\times$: the shared-leaf dedupe
became a bit mask per column block instead of a set intersection; the tile
reduction became a tree of horizontal adds with a vectorised epilogue; a
SIMD prefilter compares each tile's distances with both heaps' current worst
keys and discards, without a heap operation, offers the heap would reject
(exact, so the graph is unchanged); and the scan became symmetric, halving
the arithmetic, with a lock-scheduled round-robin tournament order replacing
barriered rounds that had left a quarter of the threads idle. Because heap
offers commute (Proposition~\ref{prop:det}), the tile size cannot change the
graph; it is chosen per leaf as the largest power of two in $[32,128]$ that
leaves at least 100 blocks, within $11\%$ of the measured optimum in every
case we swept. At $d=960$ the tile reaches only $1/7.7$ of OpenBLAS at the
same dimension (eight rows of 960 floats overflow L1 and there is no blocking
over the dimension), the largest remaining engineering headroom.

\begin{table*}[t]
\centering\small
\caption{Candidate-stage engineering, d0, 64 threads. \emph{Smoke}: 40
leaves of 25,000 points. Rates are ordered-pair equivalents; the machine's
SGEMM rate is $4.45$ G/s.}
\label{tab:kernel}
\begin{tabular}{@{}lrrr@{}}
\toprule
step & smoke & GloVe & G pairs/s \\
\midrule
one-sided scan, $2\times4$ AVX2 tile      & 29.1 s & 250.5 s & 0.71 \\
$+$ leaf masks, hadd tree, SIMD prefilter & 15.2 s & 125.3 s & 1.41 \\
symmetric scan, barrier per round         & 14.6 s & 101.7 s & 1.74 \\
$+$ two-sided prefilter, locked scheduler & 11.1 s & 77.8 s  & 2.28 \\
$+$ $TB=128$ (default)                    & \textbf{10.3 s} & \textbf{73.3 s} & \textbf{2.42} \\
\midrule
RTX 3060 Laptop (6 GB), FAISS fused $k$NN & --- & 22.2 s & \textbf{7.90} \\
\bottomrule
\end{tabular}
\end{table*}

\paragraph{On a GPU the stage is a library call} Exact $k$NN inside
independent blocks is the shape of a fused GEMM-plus-select kernel. With
FAISS's \texttt{knn\_gpu}~\cite{johnson2021gpu,faiss} on a 6 GB laptop card
and the prune and spine on the laptop's CPU, the GloVe build takes $54.2$ s
end to end against $640$ s for Vamana on the same laptop and $94.7$ s for the
CPU build on d0, at Vamana parity on Vamana's searcher ($9{,}131/13{,}447$
against $9{,}204/13{,}325$ at recall $0.95/0.97$). Vamana's per-core rate on
the laptop matches d0's exactly: a search builder is bounded by
dependent-access latency, which does not improve with core count. The GPU
$k$-means accumulates with scatter-add in unspecified order, so that path is
not bitwise reproducible.

\paragraph{A quantiser is a memory lever inside a pool} The 2-bit
quantisation of Connor, Dearle and Claydon~\cite{connor2025evp} keeps the $x$
largest $|u_i|$ of a unit vector as $\pm1$ and zeroes the rest. Globally, on
GloVe with $x = 50$, a float query against 2-bit data recovers $0.960$ of the
true 10-NN in its top 500 ($0.827$ at 100, $0.995$ at 2,000), and ternary
against ternary $0.828$. Inside a pool it is lossless: in every pool
configuration we measured on GloVe, the top 500 by this proxy recall the
pool's own exact ceiling to within one or two points. Whatever it loses, it
loses on the global problem, not inside a partition; it would cut the scan's
operand size by $16\times$, not change its quality.

\section{The spine in full}\label{app:spine}

\begin{table*}[t]
\centering\small
\caption{Clustered set (800 Gaussian clusters, $n=200{,}000$, $d=128$, \dense{}
L240/m5): recall@10 at $\mathit{ef}$ 64 / 128 on our searcher, by sample size
and degree cap. No spine $0.002$; SAT $0.987$ / $0.994$.\src{\texttt{adaptive-spine-FINDINGS.md} section 3}}
\label{tab:smoke}
\setlength{\tabcolsep}{3pt}
\begin{tabular}{@{}lrrrr@{}}
\toprule
$|S|$ & uncapped & cap 32 & cap 16 & cap 8 \\
\midrule
1k  & 0.531/0.611 & 0.531/0.611 & 0.467/0.561 & 0.237/0.310 \\
3k  & 0.865/0.949 & 0.865/0.949 & 0.756/0.903 & 0.441/0.599 \\
10k & 0.993/1.000 & 0.992/1.000 & 0.968/0.996 & 0.657/0.828 \\
30k & 1.000/1.000 & 1.000/1.000 & 0.998/1.000 & 0.808/0.918 \\
\bottomrule
\end{tabular}
\end{table*}

\begin{table*}[t]
\centering\small
\caption{Spine price on real data, \dq{0.95} / \dq{0.97}, \dense{} build, by
sample size and degree cap. No spine: GloVe $9{,}130/13{,}031$, GIST
$4{,}259/5{,}766$.\src{\texttt{adaptive-spine-FINDINGS.md} section 3}}
\label{tab:capfull}
\setlength{\tabcolsep}{3pt}
\begin{tabular}{@{}llrrrr@{}}
\toprule
& $|S|$ & uncapped & cap 32 & cap 16 & cap 8 \\
\midrule
GloVe & 1k  & 9,138/12,986 & 9,136/12,997 & 9,121/12,995 & 9,140/13,028 \\
      & 3k  & 9,165/12,998 & 9,159/13,016 & 9,184/13,042 & 9,139/13,052 \\
      & 10k & 9,337/13,216 & 9,198/13,076 & 9,172/13,057 & 9,181/13,069 \\
      & 30k & 9,420/13,707 & 9,295/13,210 & 8,961/13,117 & 9,163/13,063 \\
\midrule
GIST  & 1k  & 4,249/5,832 & 4,249/5,832 & 4,246/5,827 & 4,244/5,829 \\
      & 3k  & 4,289/5,903 & 4,290/5,900 & 4,287/5,890 & 4,300/5,888 \\
      & 10k & 4,331/5,896 & 4,296/5,854 & 4,268/5,817 & 4,268/5,822 \\
      & 30k & 4,501/6,138 & 4,330/5,896 & 4,298/5,864 & 4,283/5,799 \\
\bottomrule
\end{tabular}
\end{table*}

Tables~\ref{tab:smoke} and~\ref{tab:capfull} are the full grid behind
Section~\ref{sec:spine}. On real data the cap removes the price at every
sample size; on the clustered set it removes coverage at 8 and keeps it at
16 from $|S| = 10$k. HSP degree quantiles at $|S| = 100$k (uncapped): GloVe
median 81, p99 182, max 381; GIST median 32, p99 163, max 1,048. The
automatic rule's rounds on GIST with orphan repair stop at $|S| = 2{,}048$
after one doubling (999,918 of $10^6$ vertices reachable, 1,982 orphans
repaired); without the repair the reachable count moves by fewer than 150
vertices per doubling and the rule runs to its cap.

\section{Adaptive memberships}\label{app:adaptive}

\begin{table*}[t]
\centering\small
\caption{Two-level membership policies end to end (\dense{}, SAT spine):
the hardest fraction $f$ of points by the score gets $m_{hi}$, the rest
$m_{lo}$. \emph{A}: most cover at the uniform-$m=5$ pair budget; \emph{B}:
fewest pairs at its cover. \dq{0.95}/\dq{0.97} on ParlayANN's searcher.\src{\texttt{adaptive-spine-FINDINGS.md} section 2}}
\label{tab:adaptive}
\setlength{\tabcolsep}{4pt}
\begin{tabular}{@{}llrrr@{}}
\toprule
corpus & policy & mean $m$ & cand.\ G & \dq{0.95}/\dq{0.97} \\
\midrule
GloVe & uniform $m=5$ & 5.0 & 88.6 & 9,069/13,322 \\
      & oracle A ($f$ .3, 4/7) & 4.9 & 85.4 & 8,997/13,288 \\
      & pilot A & 4.9 & 85.4 & 9,031/13,293 \\
      & margin A & 4.9 & 84.8 & 9,046/13,403 \\
\midrule
GIST  & uniform $m=5$ & 5.0 & 26.4 & 4,485/6,055 \\
      & oracle A ($f$ .5, 4/6) & 5.0 & 26.3 & 4,430/5,956 \\
      & pilot A & 5.0 & 26.4 & 4,407/5,962 \\
      & oracle/pilot/margin B & 5.5 & 32.0--32.2 & 4,434--4,535/5,990--6,048 \\
\bottomrule
\end{tabular}
\end{table*}

The scores are the exact local clustering (oracle), the clustering of the
approximate graph a pilot $m=2$ pass produces (Spearman $0.896$ with the
oracle on GloVe, $0.677$ on GIST; its pairs are a subset of the final
build's, so it costs nothing extra), and the leader margin $d_1/d_2$. On the
sampled cover-versus-pairs frontier the best two-level policy at the uniform
budget saves $3$--$4\%$ of pairs on GloVe and nothing on GIST; aggressive
policies ($f \le 0.1$, $m_{hi}\ge 8$) lose, because a hard point's extra
leaves grow those leaves for every member (Table~\ref{tab:adaptive}).

\section{Throughput}\label{app:qps}

\begin{table*}[t]
\centering\small
\caption{QPS at recall $0.95$ ($k=10$), the runs of Table~\ref{tab:main},
ParlayANN's batch search on d0's 64 threads. \emph{hnswlib native}: its own
hierarchical search, interpolated on its $\mathit{ef}$ ladder.\src{\texttt{open-items-FINDINGS.md} (K), \texttt{framing-FINDINGS.md} (T1, T5)}}
\label{tab:qps}
\setlength{\tabcolsep}{3pt}
\begin{tabular}{@{}lrrrrrr@{}}
\toprule
corpus & \dense{} & Vamana & PiPNN & comp. & hnsw L0 & hnsw native \\
\midrule
GloVe     & 14.6k & 17.9k & 10.8k & 10.4k & 7.1k & 6.6k \\
SIFT      & 75k & 135k & 110k & 98k & 97k & 92k \\
GIST      & 2.1k & 3.3k & 3.2k & 3.0k & 1.7k & 2.8k \\
Deep-10M  & 65k & 81k & 63k & 59k & 50k & 41k \\
Wikipedia & 4.1k & 4.1k & 4.0k & 4.4k & 3.4k & 3.2k \\
Deep-100M & 38.6k & 37.8k & 33.3k & 35.2k & 32.7k & 17.1k \\
\bottomrule
\end{tabular}
\end{table*}

Vamana's graphs run faster than their distance counts suggest on GloVe,
SIFT and GIST (on GIST, $3.3$k QPS at $5{,}281$ evaluations against the dense
graph's $2.1$k at $4{,}268$); the difference closes at $d=1024$ (Wikipedia)
and at $10^8$, where the dense and composed graphs need enough fewer
evaluations to compensate (Table~\ref{tab:qps}). We have not measured the
cause (degree, memory layout). hnswlib's own hierarchical search is slower
than ParlayANN's search over its level-0 graph on five of six corpora at this
recall, so the flat grading of Table~\ref{tab:main} does not understate
hnswlib's throughput, only its distance count.

\begingroup\sloppy\hbadness=10000
\bibliographystyle{elsarticle-num}
\bibliography{dense}

\begin{thebibliography}{10}
\expandafter\ifx\csname url\endcsname\relax
  \def\url#1{\texttt{#1}}\fi
\expandafter\ifx\csname urlprefix\endcsname\relax\def\urlprefix{URL }\fi
\expandafter\ifx\csname href\endcsname\relax
  \def\href#1#2{#2} \def\path#1{#1}\fi

\bibitem{rubel2026pipnn}
T.~Rubel, R.~Wen, L.~Dhulipala, L.~Gottesb{\"u}ren, R.~Jayaram,
  J.~{\L}{\k{a}}cki, {PiPNN}: Ultra-scalable graph-based nearest neighbor
  indexing, arXiv:2602.21247 (2026).

\bibitem{malkov2020hnsw}
Y.~A. Malkov, D.~A. Yashunin, Efficient and robust approximate nearest neighbor
  search using hierarchical navigable small world graphs, IEEE TPAMI 42~(4)
  (2020) 824--836.

\bibitem{fu2019nsg}
C.~Fu, C.~Xiang, C.~Wang, D.~Cai, Fast approximate nearest neighbor search with
  the navigating spreading-out graph, PVLDB 12~(5) (2019) 461--474.

\bibitem{subramanya2019diskann}
S.~J. Subramanya, Devvrit, H.~V. Simhadri, R.~Krishnaswamy, R.~Kadekodi,
  {DiskANN}: Fast accurate billion-point nearest neighbor search on a single
  node, in: NeurIPS, 2019, pp. 13748--13758.

\bibitem{manohar2024parlayann}
M.~D. Manohar, Z.~Shen, G.~E. Blelloch, L.~Dhulipala, Y.~Gu, H.~V. Simhadri,
  Y.~Sun, {ParlayANN}: Scalable and deterministic parallel graph-based
  approximate nearest neighbor search algorithms, in: PPoPP, 2024, pp.
  270--285.

\bibitem{chavez2026graft}
E.~Ch\'avez, {Graft}: Navigable graphs grown on spatial-approximation
  rootstock, Information Systems, to appear (2026).

\bibitem{chavez2026wall}
E.~Ch\'avez, Batched feedback and the random-access wall in search-based graph
  construction, arXiv:2609.30493 (2026).

\bibitem{chen2009recursive}
J.~Chen, H.~ren Fang, Y.~Saad, Fast approximate {kNN} graph construction for
  high dimensional data via recursive {Lanczos} bisection, Journal of Machine
  Learning Research 10 (2009) 1989--2012.

\bibitem{wang2012scalable}
J.~Wang, J.~Wang, G.~Zeng, Z.~Tu, R.~Gan, S.~Li, Scalable $k$-{NN} graph
  construction for visual descriptors, in: Proc.\ IEEE Conference on Computer
  Vision and Pattern Recognition (CVPR), 2012.

\bibitem{munoz2019hcnng}
J.~A. Vargas~Mu{\~n}oz, M.~A. Gon{\c{c}}alves, Z.~Dias, R.~da~Silva~Torres,
  Hierarchical clustering-based graphs for large scale approximate nearest
  neighbor search, Pattern Recognition 96 (2019) 106970.

\bibitem{groh2023ggnn}
F.~Groh, L.~Ruppert, P.~Wieschollek, H.~P.~A. Lensch, {GGNN}: Graph-based {GPU}
  nearest neighbor search, IEEE Transactions on Big Data 9~(1) (2023).

\bibitem{ootomo2024cagra}
H.~Ootomo, A.~Naruse, C.~Nolet, R.~Wang, T.~Feher, Y.~Wang, {CAGRA:} highly
  parallel graph construction and approximate nearest neighbor search for gpus,
  in: 40th {IEEE} International Conference on Data Engineering, {ICDE} 2024,
  Utrecht, The Netherlands, May 13-16, 2024, {IEEE}, 2024, pp. 4236--4247.

\bibitem{cuvs}
{NVIDIA}, {cuVS}: {GPU}-accelerated vector search and clustering (all-neighbors
  $k$-nn graph construction), \url{https://github.com/rapidsai/cuvs} (2025).

\bibitem{li2025tagore}
Z.~Li, X.~Ke, Y.~Zhu, B.~Yu, B.~Zheng, Y.~Gao, Scalable graph indexing using
  {GPUs} for approximate nearest neighbor search, Proceedings of the ACM on
  Management of DataArXiv:2508.08744 (2025).
\newblock \href {https://doi.org/10.1145/3769825} {\path{doi:10.1145/3769825}}.

\bibitem{lloyd1982}
S.~P. Lloyd, Least squares quantization in {PCM}, IEEE Transactions on
  Information Theory 28~(2) (1982) 129--137.

\bibitem{navarro2002sat}
G.~Navarro, Searching in metric spaces by spatial approximation, The VLDB
  Journal 11~(1) (2002) 28--46.

\bibitem{pennington2014glove}
J.~Pennington, R.~Socher, C.~D. Manning, Glove: Global vectors for word
  representation, in: A.~Moschitti, B.~Pang, W.~Daelemans (Eds.), Proceedings
  of the 2014 Conference on Empirical Methods in Natural Language Processing,
  {EMNLP} 2014, October 25-29, 2014, Doha, Qatar, {A} meeting of SIGDAT, a
  Special Interest Group of the {ACL}, {ACL}, 2014, pp. 1532--1543.

\bibitem{jegou2011pq}
H.~J{\'{e}}gou, M.~Douze, C.~Schmid, Product quantization for nearest neighbor
  search, {IEEE} Trans. Pattern Anal. Mach. Intell. 33~(1) (2011) 117--128.

\bibitem{babenko2016deep1b}
A.~Babenko, V.~Lempitsky, Efficient indexing of billion-scale datasets of deep
  descriptors, in: IEEE Conference on Computer Vision and Pattern Recognition
  (CVPR), 2016, pp. 2055--2063.

\bibitem{chavez2005hsp}
E.~Ch{\'{a}}vez, S.~Dobrev, E.~Kranakis, J.~Opatrny, L.~Stacho, H.~Tejeda,
  J.~Urrutia, Half-space proximal: {A} new local test for extracting a bounded
  dilation spanner of a unit disk graph, in: J.~H. Anderson, G.~Prencipe,
  R.~Wattenhofer (Eds.), Principles of Distributed Systems, 9th International
  Conference, {OPODIS} 2005, Pisa, Italy, December 12-14, 2005, Revised
  Selected Papers, Vol. 3974 of Lecture Notes in Computer Science, Springer,
  2005, pp. 235--245.

\bibitem{malkov2014nsw}
Y.~Malkov, A.~Ponomarenko, A.~Logvinov, V.~Krylov, Approximate nearest neighbor
  algorithm based on navigable small world graphs, Inf. Syst. 45 (2014) 61--68.

\bibitem{harwood2016fanng}
B.~Harwood, T.~Drummond, {FANNG}: Fast approximate nearest neighbour graphs,
  in: CVPR, 2016, pp. 5713--5722.

\bibitem{hong1981io}
J.-W. Hong, H.~T. Kung, {I/O} complexity: The red-blue pebble game, in: STOC,
  1981, pp. 326--333.

\bibitem{williams2009roofline}
S.~Williams, A.~Waterman, D.~Patterson, Roofline: An insightful visual
  performance model for multicore architectures, Communications of the ACM
  52~(4) (2009) 65--76.

\bibitem{gottesburen2025partition}
L.~Gottesb{\"u}ren, L.~Dhulipala, R.~Jayaram, J.~{\L}{\k{a}}cki, Unleashing
  graph partitioning for large-scale nearest neighbor search, Proceedings of
  the VLDB Endowment 18 (2025).

\bibitem{chen2021spann}
Q.~Chen, B.~Zhao, H.~Wang, M.~Li, C.~Liu, Z.~Li, M.~Yang, J.~Wang, {SPANN}:
  Highly-efficient billion-scale approximate nearest neighbor search, in:
  Advances in Neural Information Processing Systems (NeurIPS), 2021, pp.
  5199--5212.

\bibitem{johnson2021gpu}
J.~Johnson, M.~Douze, H.~J{\'e}gou, Billion-scale similarity search with
  {GPUs}, IEEE Transactions on Big Data 7~(3) (2021) 535--547.

\bibitem{faiss}
M.~Douze, A.~Guzhva, C.~Deng, J.~Johnson, G.~Szilvasy, P.-E. Mazar\'e,
  M.~Lomeli, L.~Hosseini, H.~J\'egou, The faiss library, arXiv:2401.08281
  (2024).

\bibitem{dong2011nndescent}
W.~Dong, M.~Charikar, K.~Li, Efficient k-nearest neighbor graph construction
  for generic similarity measures, in: S.~Srinivasan, K.~Ramamritham, A.~Kumar,
  M.~P. Ravindra, E.~Bertino, R.~Kumar (Eds.), Proceedings of the 20th
  International Conference on World Wide Web, {WWW} 2011, Hyderabad, India,
  March 28 - April 1, 2011, {ACM}, 2011, pp. 577--586.

\bibitem{fu2016efanna}
C.~Fu, D.~Cai, {EFANNA} : An extremely fast approximate nearest neighbor search
  algorithm based on knn graph, arXiv:1609.07228 (2016).

\bibitem{sptag2018}
Q.~Chen, H.~Wang, M.~Li, G.~Ren, S.~Li, J.~Zhu, J.~Li, C.~Liu, L.~Zhang,
  J.~Wang, \href{https://github.com/Microsoft/SPTAG}{{SPTAG}: A library for
  fast approximate nearest neighbor search}, GitHub (2018).
\newline\urlprefix\url{https://github.com/Microsoft/SPTAG}

\bibitem{ono2023rnndescent}
N.~Ono, Y.~Matsui, Relative {NN-Descent}: A fast index construction for
  graph-based approximate nearest neighbor search, in: Proc.\ 31st ACM
  International Conference on Multimedia (MM), 2023.

\bibitem{yang2025revisiting}
S.~Yang, J.~Xie, Y.~Liu, J.~X. Yu, X.~Gao, Q.~Wang, Y.~Peng, J.~Cui, Revisiting
  the index construction of proximity graph-based approximate nearest neighbor
  search, Proceedings of the VLDB Endowment 18~(6) (2025) 1825--1838.

\bibitem{toussaint1980rng}
G.~T. Toussaint, The relative neighbourhood graph of a finite planar set,
  Pattern Recognit. 12~(4) (1980) 261--268.

\bibitem{arya1993rng}
S.~Arya, D.~M. Mount, Approximate nearest neighbor queries in fixed dimensions,
  in: SODA, 1993, pp. 271--280.

\bibitem{aumuller2021lid}
M.~Aum{\"u}ller, M.~Ceccarello, The role of local dimensionality measures in
  benchmarking nearest neighbor search, Information Systems 101 (2021) 101807.

\bibitem{chavez2001searching}
E.~Ch{\'{a}}vez, G.~Navarro, R.~A. Baeza{-}Yates, J.~L. Marroqu{\'{\i}}n,
  Searching in metric spaces, {ACM} Comput. Surv. 33~(3) (2001) 273--321.

\bibitem{guo2020scann}
R.~Guo, P.~Sun, E.~Lindgren, Q.~Geng, D.~Simcha, F.~Chern, S.~Kumar,
  Accelerating large-scale inference with anisotropic vector quantization, in:
  International Conference on Machine Learning (ICML), 2020, pp. 3887--3896.

\bibitem{gao2024rabitq}
J.~Gao, C.~Long, {RaBitQ}: Quantizing high-dimensional vectors with a
  theoretical error bound for approximate nearest neighbor search, Proceedings
  of the ACM on Management of Data 2~(3) (2024) 167.

\bibitem{connor2025evp}
R.~Connor, A.~Dearle, B.~Claydon, Equi-{V}oronoi polytopes: A geometric basis
  for 2-bit quantisation, in: Similarity Search and Applications (SISAP 2025),
  Vol. 16134 of Lecture Notes in Computer Science, Springer, 2025, pp.
  288--302.
\newblock \href {https://doi.org/10.1007/978-3-032-06069-3_23}
  {\path{doi:10.1007/978-3-032-06069-3_23}}.

\end{thebibliography}
\endgroup

\end{document}